\documentclass[journal, draftclsnofoot,oneside, onecolumn,11pt, letterpaper]{IEEEtran}

\ifCLASSINFOpdf
\else
\fi
	\usepackage{pslatex}
	\usepackage{amsfonts,color,morefloats}
	\usepackage{amssymb,amsmath,latexsym,amsthm}
	\usepackage{bbm}
	\usepackage{amsfonts}
	\usepackage{mathrsfs}
	\usepackage{amsthm}
	\usepackage{latexsym}
	\usepackage{color}
	\usepackage{dcolumn}
	\usepackage{multicol}
	\usepackage{extarrows}
	\usepackage{amsmath}
	
	\allowdisplaybreaks[4]

	\newtheorem{theorem}{Theorem}
	\newtheorem{lemma}[theorem]{Lemma}

	\newtheorem{definition}{Definition}
	\newtheorem{example}{Example}

\begin{document}
		%
		\title{The 4-Adic Complexity of Interleaved  Quaternary Sequences of Even Length with Optimal Autocorrelation
		}
		%
		%
		%
		
		\author{Xiaoyan~Jing,
			Zhefeng~Xu,
			Minghui~Yang,
			and~Keqin~Feng
			\thanks{Xiaoyan Jing is with Research Center for Number Theory and Its Applications, Northwest University, Xi'an 710127, China (Email: jxymg@126.com).}
			\thanks{ Zhefeng Xu is with Research Center for Number Theory and Its Applications, Northwest University, Xi'an 710127, China (Email: zfxu@nwu.edu.cn).}
			\thanks{Minghui Yang is with State Key Laboratory of Information Security, Institute of Information Engineering, Chinese Academy of Sciences, Beijing 100093, China (e-mail:  yangminghui6688@163.com).}
			\thanks{Keqin Feng is with the department of Mathematical Sciences, Tsinghua University, Beijing 100084, China (email: fengkq@tsinghua.edu.cn).}
			 \thanks{This work is supported by the State Key Program of National Natural Science Foundation of China (12031011) and the National Natural Science Foundation of China (11701553, 11971381).}
		}
		
		%
		%

	\markboth{}%
	{Shell \MakeLowercase{\textit{et al.}}: Bare Demo of IEEEtran.cls for IEEE Journals}
	%



	\maketitle
	

	%
	\IEEEpeerreviewmaketitle
	
	\begin{abstract}
		Su et al. proposed several new classes of quaternary sequences of even length  with optimal
		autocorrelation interleaved by twin-prime sequences pairs, GMW sequences pairs or binary cyclotomic sequences of order four in \cite{S1}. In this paper, we determine the 4-adic complexity of these quaternary sequences with period $2n$ by using correlation function  and the ``Gauss periods" of order four and ``quadratic Gauss sums" on finite field $\mathbb{F}_n$ and valued in  $\mathbb{Z}^{*}_{4^{2n}-1}$. Our results show that they are safe enough to resist the attack of the rational approximation algorithm.
	\end{abstract}
	
	\begin{IEEEkeywords}
		quaternary sequences, 4-adic complexity, interleaved sequences, optimal autocorrelation, ``quadratic Gauss sums"
	\end{IEEEkeywords}

	%
	\IEEEpeerreviewmaketitle

	\section{Introduction}\label{section1}
	Binary and quaternary sequences play important roles in communication and cryptography systems. They are expected to have good autocorrelation and high complexity (linear complexity, $N$-adic complexity and so on) for high speed and security of communication. For two $N$-ary sequences $s=\{s_i\}_{i=0}^{n-1}$ and $t=\{t_i\}_{i=0}^{n-1}$ with period $n$, the correlation function of $s$ and $t$ is defined by
	$$R_{s,t}(\tau)=\sum_{i=0}^{n-1}\xi^{s_{i+\tau}-t_i}\in {\mathbb{Z}[\xi]},\ 0\leq\tau<n,$$
	where $\xi$ is a primitive $N$-th root of unity. When $t=s$, the correlation function is called autocorrelation function of sequence $s$ and denoted by $R_{s}(\tau)$.  The maximum out-of-phase autocorrelation magnitude of $s$ is defined as
	$R_{max}(s) = \max\{|R_s(\tau)| : 1\leq \tau < n\}$.
	In practice, we need the out-of-phase autocorrelation magnitude $R_{max}(s)$ to be as small as possible. Specially,
	for a  quaternary sequence with even period $2n$, it is called optimal  autocorrelation sequence if $R_{max}(s)=2$ \cite{T1}.
	Several optimal quaternary sequences have been given by \cite{Kim1}, \cite{L}, \cite{Kim2}, \cite{J}. Thereafter Su et al. \cite{S1} presented several new families of optimal quaternary sequences constructed by interleaving operater,  twin-prime sequences pairs and GMW sequences pairs given by Tang and Gong in 2010 \cite{T2} or binary cyclotomic sequences of order four. For a optimal quaternary sequence, it is hoped that the 4-adic complexity be as large as possible.
	
	For a $N$-ary sequence $s=\{s_i\}_{i=0}^{n-1}$ with period $n$,
	the N-adic complexity of $s$ is defined by
	$$C_{N}(s)=\log_{N}\frac{N^{n}-1}{d}, $$
	where $d=\gcd (S(N), N^{n}-1)$, $S(N)=\sum_{i=0}^{n-1}s_{i}N^{i}\in \mathbb{Z}$.
	The $N$-adic complexity $C_N(s)$ measures that the smallest length of the feedback with carry shift register (FCSR) to generate an $N$-ary sequence.  To resist the attack of the rational approximation algorithm, the 4-adic complexity of a quaternary sequence $s$ with period $2n$ should exceed $\frac{2n-16}{6}$  \cite{K1}, \cite{K2}. The 4-adic complexity of several quaternary sequences has been computed in \cite{Q1}, \cite{Y}, \cite{Q}, \cite{EA.}, we will compute the 4-adic complexity of the optimal quaternary sequences given by Su \cite{S1} and show that the 4-adic complexity of such sequences is large enough to resist the attack of the rational approximation algorithm in this paper.

	In Section \ref{sec2} we briefly introduce the interleaved quaternary sequences constructed by binary sequences pair with the same period $n$, and present the optimal quaternary sequences given by Su in \cite{S1} from a pair of twin-prime sequences, GMW sequences or binary cyclotomic sequences of order four. In Section \ref{sec3} we give the calculation of the 4-adic complexity of quaternary sequences interleaved by a pair of twin-prime sequences or GMW sequences using correlation function  and give some  examples to demonstrate the main results. In Section \ref{sec4} we present the ``Gauss periods" of order four and  ``quadratic Gauss sums" on finite field $\mathbb{F}_n$ and  valued in  $\mathbb{Z}^{*}_{4^{2n}-1}$, then use them to determine the 4-adic complexity of interleaved quaternary sequences constructed by two or three binary cyclotomic sequences of order four. In Section \ref{sec3} we give the calculation of the 4-adic complexity of quaternary sequenc of order four. We also give some examples to verify the correctness of our results. Section \ref{sec5} is a conclusion of this paper.

	\section{Interleaved Quaternary Sequences}\label{sec2}
	In this section, we give a brief introduction of interleaved quaternary sequences with period $2n$.
	Let $n$ be an odd integer, $n\geq3$, $c^0=\{c^0_j\}_{j=0}^{n-1}$ and $c^1=\{c^1_j\}_{j=0}^{n-1}$, $c^2=\{c^2_j\}_{j=0}^{n-1}$ and $c^3=\{c^3_j\}_{j=0}^{n-1}$ be two pairs of binary sequences with the same period $n$, \ $\lambda=\frac{n+1}{2},\ e=(e_0, e_1, e_2)$ be a binary sequence defined over $\mathbb{F}_2$. Considering the following two $n\times2$ matrices over $\mathbb{F}_{2}=\{0, 1\}$:
	$$(a_{(l,k)})_{n\times2}=\left(
	\begin{array}{cc}
		c^0_0 & c^1_{0+\lambda}\oplus e_0 \\
		c^0_1 & c^1_{1+\lambda}\oplus e_0  \\
		\vdots & \vdots  \\
		c^0_{n-1} & c^1_{n-1+\lambda}\oplus e_0   \\
	\end{array}
	\right),
	\ \ (b_{(l,k)})_{n\times2}=\left(
	\begin{array}{cc}
		c^2_0\oplus e_1 & c^3_{0+\lambda}\oplus e_2 \\
		c^2_1\oplus e_1 & c^3_{1+\lambda}\oplus e_2 \\
		\vdots & \vdots  \\
		c^2_{n-1}\oplus e_1 & c^3_{n-1+\lambda}\oplus e_2 \\
	\end{array}
	\right),
	$$
	where $\oplus$ denotes binary addition over $\mathbb{F}_2$.
	
	Then we have two interleaved binary sequences $a(2l+k)=a_{(l,k)}$ and $b(2l+k)=b_{(l,k)}$ with period $2n$ from the matrix $(a_{(l,k)})_{n\times2}$ and $(b_{(l,k)})_{n\times2}$ respectively, where $0\leq l<n,\ 0\leq k<2$. We denote $$a=\emph{I}(c^0, L^{\lambda}(c^1)\oplus e_0),\ b=\emph{I}(c^2\oplus e_1, L^{\lambda}(c^3)\oplus e_2)$$
	for convenience, where $\emph{I}$ is the interleaving operator and $L^{\lambda}(c)$ is $\lambda$-shift of the sequence $c$, namely
	\begin{align*}
		&a=(c^0_0, c^1_{0+\lambda}\oplus e_0, c^0_1, c^1_{1+\lambda}\oplus e_0, \cdots,  c^0_{n-1}, c^1_{n-1+\lambda}\oplus e_0),\\
		&b=(c^2_0\oplus e_1, c^3_{0+\lambda}\oplus e_2, c^2_1\oplus e_1, c^3_{1+\lambda}\oplus e_2, \cdots, c^2_{n-1}\oplus e_1, c^3_{n-1+\lambda}\oplus e_2).
	\end{align*}
	
	It is well known that the Gray mapping is defined as
	$$\phi: \mathbb{F}_{2}\times \mathbb{F}_{2}\cong \mathbb{Z}_{4},$$
	$$\phi(0, 0)=0,~ \phi(0, 1)=1,~ \phi(1, 1)=2, ~\phi(1, 0)=3.$$
	Using  Gray mapping $\phi$, a family of quaternary sequences $s$ with period $2n$ is constructed by Su \cite{S1}:
	\begin{align}\label{E1}
		s=s(a, b)=\{s_i\}_{i=0}^{2n-1}, ~~s_{i}=\phi(a_{i}, b_{i}).
	\end{align}
	
		\subsection{Quaternary sequences with optimal autocorrelation constructed by sequences pair}
	
	It is proved that if $c^0$ and $c^1$, $c^2$ and $c^3$ are two twin-prime sequences pairs or two GMW sequences pairs with period $n$ given in \cite{T2}, the interleaved quaternary sequences $s=s(a, b)$ have optimal autocorrelation $R_{s}(\tau)\in \{0, -2\}$ for all $1\leq \tau\leq 2n-1$.
	
	
	
	Then we have the following lemma.
	
	\begin{lemma}\label{lem1}(\cite{S1})
		Let $t_0$ and $t_1$ be the twin-prime sequences pair of length $n=p(p+2)$ or GMW sequences pair of length $n=2^{2k}-1$. If $e=(e_0,e_1,e_2)$ satisfies $e_0+e_1+e_2\equiv1\pmod{2}$ and
		$$(c^0,c^1,c^2,c^3)\in\{(t_0,t_1,t_0,t_1),(t_0,t_1,t_1,t_0),(t_1,t_0,t_1,t_0),(t_1,t_0,t_0,t_1)\}.$$
		Then the quaternary sequence $s$ given by Construction (1) is an optimal autocorrelation sequence.
	\end{lemma}
	
	For a sequence $t$ with period $n$, it is called ideal autocorrelation sequences if the autocorrelation function
	\begin{align}\label{ideal}
		R_t(\tau)=\left\{ \begin{array}{ll}
			n,\ \text{if}\ \tau=0\\
			-1,\ \text{otherwise}. 
		\end{array}\right.
	\end{align}
	And as we all know that twin-prime sequences and GMW sequences are ideal autocorrelation sequences.  The autocorrelation of modified twin-prime sequences, modified GMW sequences, the cross-correlation between twin-prime sequences and modified twin-prime sequences, GMW sequences and modified GMW sequences were determinded by Tang in \cite{T2}, for the convenience of using in the rest part, we list the results as the following lemmas.
	
		\begin{lemma}\label{lem2}(\cite{T2})
		Let $t_0$ and $t_1$ be the twin-prime sequence and modified twin-prime sequence of length $p(p+2)$ respectiely, for $0\leq\tau<p(p+2)$, we have
		\begin{align*}
				R_{t_1}(\tau)
			=\left\{ \begin{array}{ll}
				p(p+2),\ & \text{if}\ \tau=0\\
				-1,\ &\text{if}\ \tau=0\pmod{p+2}\ \text{and} \ \tau\neq0\\
				3,\ & \text{otherwise}
			\end{array} \right.,
		\end{align*}
	and 
	\begin{align*}
		R_{t_0,t_1}(\tau)=R_{t_1,t_0}(\tau)
		=\left\{ \begin{array}{ll}
			p^2,\ & \text{if}\ \tau=0\\
			-2p-1,\ &\text{if}\ \tau=0\pmod{p+2}\ \text{and} \ \tau\neq0\\
			1,\ & \text{otherwise}
		\end{array} \right..
	\end{align*}
	\end{lemma}
			\begin{lemma}\label{lem3}(\cite{T2})
		Let $t_0$ and $t_1$ be the GMW sequence and modified GMW sequence of length $2^{2k}-1$ respectiely,  for $0\leq\tau<2^{2k}-1$, we have
		\begin{align*}
			R_{t_1}(\tau)
			=\left\{ \begin{array}{ll}
				2^{2k}-1,\ & \text{if}\ \tau=0\\
				-1,\ &\text{if}\ \tau=0\pmod{2^k+1}\ \text{and} \ \tau\neq0\\
				3,\ & \text{otherwise}
			\end{array} \right.,
		\end{align*}
		and 
		\begin{align*}
			R_{t_0,t_1}(\tau)=R_{t_1,t_0}(\tau)
			=\left\{ \begin{array}{ll}
				2^{2k}-2^{k+1}+1,\ & \text{if}\ \tau=0\\
				-2^{k+1}+1,\ &\text{if}\ \tau=0\pmod{2^k+1}\ \text{and} \ \tau\neq0\\
				1,\ & \text{otherwise}
			\end{array} \right..
		\end{align*}
	\end{lemma}

	\subsection{Quaternary sequences with optimal autocorrelation constructed by binary cyclotomic sequences of order four}

Let $n=4f+1$ be a prime number, $\mathbb{F}_n^{\ast}=\langle\theta\rangle$, $C=\langle\theta^4\rangle$ and $D_\gamma=\theta^\gamma C \ (0\leq\gamma\leq3)$ be the cyclotomic classes of order four in $\mathbb{F}_n$.
\begin{definition}\label{def2}
	The cyclotomic numbers of order four in $\mathbb{F}_n$ are defined by, for $0\leq i,j\leq 3$
	$$(i, j)=|(D_i+1)\cap D_j|=\sharp \{(\alpha,\beta): \alpha\in D_i, \beta\in D_j, \alpha+1=\beta\}.$$
\end{definition}

Let $t_i\ (1\leq i\leq 6)$ be six binary sequences of length $n$ with support sets $D_0\cap D_1$, $D_0\cap D_2$, $D_0\cap D_3$, $D_1\cap D_2$, $D_1\cap D_3$, $D_2\cap D_3$, respectively.
Namely,

\begin{minipage}{.5\linewidth}
	\begin{align*}
		t_1=\left\{ \begin{array}{ll}
			1,\ i\in D_0\cup D_1\\
			0,\ i\in D_2\cup D_3\cup \{0\}	\end{array}\right.,\\
		t_2=	\left\{ \begin{array}{ll}
			1,\ i\in D_0\cup D_2\\
			0,\ i\in D_1\cup D_3\cup \{0\}	\end{array}\right.,\\
		t_3=	\left\{ \begin{array}{ll}
			1,\ i\in D_0\cup D_3\\
			0,\ i\in D_1\cup D_2\cup \{0\}	\end{array}\right.,
	\end{align*}
\end{minipage}
\begin{minipage}{.35\linewidth}
	\begin{align*}
		t_4=\left\{ \begin{array}{ll}
			1,\ i\in D_1\cup D_2\\
			0,\ i\in D_0\cup D_3\cup \{0\}	\end{array}\right.,\\
		t_5=	\left\{ \begin{array}{ll}
			1,\ i\in D_1\cup D_3\\
			0,\ i\in D_0\cup D_2\cup \{0\}	\end{array}\right.,\\
		t_6=	\left\{ \begin{array}{ll}
			1,\ i\in D_2\cup D_3\\
			0,\ i\in D_0\cup D_1\cup \{0\}	\end{array}\right..
	\end{align*}
\end{minipage}
\bigskip

	It is proved by Su in \cite{S1} that if $c^l\in\{t_1,t_2,t_3,t_4,t_5,t_6\}$ ($l\in\{0,1,2,3\}$), the interleaved quaternary sequences with optimal autocorrelation can be constructed by $c$ and binary sequence $e=(e_0, e_1, e_2)$. The results  are shown in the following lemmas.

\begin{lemma}\label{lem7}(\cite{S1})
	Let $n=4f+1=x^2+4y^2$, $f$ be odd and $y=-1$. If $e=(e_0,e_1,e_2)$ satisfies $e_0+e_1+e_2\equiv0\pmod{2}$ and
	\begin{align*}
		(c^0,c^1,c^2,c^3)
		\in\left\{ \begin{array}{ll}
			(t_2,t_1,t_2,t_1),(t_1,t_2,t_1,t_2),(t_6,t_2,t_6,t_2),(t_2,t_6,t_2,t_6),\\ (t_5,t_4,t_5,t_4),(t_4,t_5,t_4,t_5),(t_3,t_5,t_3,t_5),(t_5,t_3,t_5,t_3)
		\end{array} \right\}.
	\end{align*}
Then the quaternary sequence $s$ given by Construction (1) is an optimal autocorrelation sequence.
\end{lemma}

\begin{lemma}\label{lem8}(\cite{S1})
	Let $n=4f+1=x^2+4y^2$, $f$ be odd and $y=-1$. If $e=(e_0,e_1,e_2)$ satisfies $e_0+e_1+e_2\equiv0\pmod{2}$ and
	\begin{align*}
		(c^0,c^1,c^2,c^3)
		\in\left\{ \begin{array}{ll}
			(t_1,t_2,t_2,t_1),(t_2,t_1,t_1,t_2),(t_2,t_6,t_6,t_2),(t_6,t_2,t_2,t_6),\\ (t_4,t_5,t_5,t_4),(t_5,t_4,t_4,t_5),(t_5,t_3,t_3,t_5),(t_3,t_5,t_5,t_3)
		\end{array} \right\}.
	\end{align*}
	Then the quaternary sequence $s$ given by Construction (1) is an optimal autocorrelation sequence.
\end{lemma}

\begin{lemma}\label{lem9}(\cite{S1})
	Let $n=4f+1=x^2+4y^2$, $f$ be odd and $y=-1$. If $e=(e_0,e_1,e_2)$ satisfies $e_0+e_1+e_2\equiv1\pmod{2}$ and
	\begin{align*}
		(c^0,c^1,c^2,c^3)
		\in\left\{ \begin{array}{ll}
			(t_2,t_1,t_6,t_2),(t_2,t_6,t_1,t_2),(t_5,t_3,t_4,t_5),(t_5,t_4,t_3,t_5),\\ (t_6,t_2,t_2,t_1),(t_1,t_2,t_2,t_6),(t_3,t_5,t_5,t_4),(t_4,t_5,t_5,t_3).
		\end{array} \right\}
	\end{align*}
Then the quaternary sequence $s$ given by Construction (1) is an optimal autocorrelation sequence.
\end{lemma}

	\section{The 4-Adic Complexity of optimal autocorrelation quaternary sequence $s$ constructed by sequences pair}\label{sec3}
	In this section, we will compute the 4-adic complexity of the quaternary sequence $s$ given by Lemma 1.
	\subsection{Condition \uppercase\expandafter{\romannumeral1}: With the notations as before, for $t_0$ and $t_1$ are the twin-prime sequences pair of length $p(p+2)$ or GMW sequences pair of length $2^{2k}-1$, and $c^0=c^2, c^1=c^3$.}
	From $c^0=c^2, c^1=c^3$, we have $a_{2j}=c_j^0$, $a_{2j+1}\equiv c_{j+\lambda}^1+e_0\pmod2$ , $b_{2j}\equiv a_{2j}+e_1\pmod 2,$ and$ \ b_{2j+1}\equiv a_{2j+1}+e_2-e_0\pmod2$.

	With $e=(1, 0, 0)$, $\lambda=\frac{n+1}{2}$, from the definition of $s$ and $S(N)$, we have
	\begin{align}\label{ES1}
		S(4)\notag
		=&\sum_{i=0}^{N-1}s_i4^i\pmod{4^{2n}-1}\notag\\ =&\sum_{i=0}^{2n-1}\phi(a_i,b_i)4^i \pmod{4^{2n}-1}\notag\\
		=&\sum_{\substack{i=0\\a_{i}=0,b_{i}=1}}^{2n-1}4^i+\sum_{\substack{i=0\\a_{i}=1,b_{i}=1}}^{2n-1}2\cdot4^i+\sum_{\substack{i=0\\a_{i}=1, b_{i}=0}}^{2n-1}3\cdot4^i
		=2\sum_{\substack{i=0\\a_{i}=1}}^{2n-1}4^i+\sum_{\substack{i=0\\a_{i}+b_{i}=1}}^{2n-1}4^i\pmod{4^{2n}-1}\notag\\
		=&2(\sum_{\substack{j=0\\a_{2j}=1}}^{n-1}4^{2j}+\sum_{\substack{j=0\\a_{2j+1}=1}}^{n-1}4^{2j+1})+
		(\sum_{\substack{j=0\\a_{2j}+b_{2j}=1}}^{n-1}4^{2j}+\sum_{\substack{j=0\\a_{2j+1}+b_{2j+1}=1}}^{n-1}4^{2j+1})\pmod{4^{2n}-1}\notag\\
		=&2(\sum_{\substack{j=0\\a_{2j}=1}}^{n-1}4^{2j}+\sum_{\substack{j=0\\a_{2j+1}=1}}^{n-1}4^{2j+1})+
		(\sum_{\substack{j=0\\a_{2j}+a_{2j}=1}}^{n-1}4^{2j}+\sum_{\substack{j=0\\a_{2j+1}+a_{2j+1}-1\equiv1\pmod2}}^{n-1}4^{2j+1})\pmod{4^{2n}-1}\notag\\
		=&2(\sum_{j=0}^{n-1}a_{2j}4^{2j}+\sum_{j=0}^{n-1}a_{2j+1}4^{2j+1})+\sum_{j=0}^{n-1}4^{2j+1}\pmod{4^{2n}-1}\notag\\
		=&2\sum_{j=0}^{n-1}c_j^04^{2j}+2\sum_{j=0}^{n-1}(1-c_{j+\lambda}^1)4^{2j+1}+\sum_{j=0}^{n-1}4^{2j+1}\pmod{4^{2n}-1}\notag\\
		=&2\sum_{j=0}^{n-1}c_j^04^{2j}-2\cdot 4^n\sum_{j=0}^{n-1}c_{j}^14^{2j}+3\cdot\sum_{j=0}^{n-1}4^{2j+1}\pmod{4^{2n}-1}.
	\end{align}
Similarly, we obtain
	\begin{align*}
	S(4)\equiv\left\{ \begin{array}{ll}
		2\sum\limits_{j=0}\limits^{n-1}c_j^04^{2j}+2\cdot 4^n\sum\limits_{j=0}\limits^{n-1}c_{j}^14^{2j}+\sum\limits_{j=0}\limits^{n-1}4^{2j}&\pmod{4^{2n}-1},\  \text{if}\  e=(0,1,0)\\
		2\sum\limits_{j=0}\limits^{n-1}c_j^04^{2j}+2\cdot 4^n\sum\limits_{j=0}\limits^{n-1}c_{j}^14^{2j}+\sum\limits_{j=0}\limits^{n-1}4^{2j+1}&\pmod{4^{2n}-1},\ \text{if}\ e=(0,0,1)\\
		2\sum\limits_{j=0}\limits^{n-1}c_j^04^{2j}-2\cdot 4^n\sum\limits_{j=0}\limits^{n-1}c_{j}^14^{2j}+9\cdot\sum\limits_{j=0}\limits^{n-1}4^{2j}&\pmod{4^{2n}-1},\  \text{if} \ e=(1,1,1)
	\end{array} \right..
\end{align*}
	\begin{theorem}\label{th1}
		Let $t_0$ and $t_1$ be the twin-prime sequences pair with period $n=p(p+2)$, and $s=s(a,b)$ be the optimal quaternary sequence given by Lemma \ref{lem1}, then for $c^0=c^2, c^1=c^3$, the 4-adic complexity of the sequence $s$ is
		\begin{align*}
			C_{4}(s)
			=\left\{ \begin{array}{ll}
				\log_{4}(4^{p(p+2)}+1)(4^{p+2}-1),\ & \text{if}\ e=(1,0,0)\ \text{or} \ e=(1,1,1)\\
				\log_{4}(4^{p(p+2)}-1)(4^{p+2}+1),\ &\text{if}\ e=(0,1,0)\ \text{or} \ e=(0,0,1)
			\end{array} \right..
		\end{align*}
	\end{theorem}
	\begin{proof}
		Assume that $c^0$ and $c^1$ are the twin-prime sequences pair of length $n=p(p+2)$, we have
		\begin{align}\label{p2}
			\sum_{j=0}^{n-1}(c_j^0-c_{j}^1)4^{2j}=\sum_{\substack{j=0\\j\equiv0\pmod{p+2}}}^{n-1}\varepsilon_1\cdot4^{2j}=\varepsilon_1\cdot\sum_{j=0}^{p-1}4^{2(p+2)j}=\varepsilon_1\cdot\frac{4^{2p(p+2)}-1}{4^{2(p+2)}-1},
		\end{align}
	where \begin{align*}
	\varepsilon_1=	\left\{ \begin{array}{ll}
			-1,\ &\text{if  $c^0$ is the twin-prime sequence and $c^1$ is the modified twin-prime sequence}\\1, \ &\text{if  $c^0$ is the modified twin-prime sequence and $c^1$ is the  twin-prime sequence}
		\end{array}\right..
	\end{align*}
	For $e=(1,\ 0,\ 0)$, from (\ref{ES1}) we know that
		\begin{align}\label{ES1.2}
			S_1(4)=&2\sum_{j=0}^{n-1}c_j^04^{2j}-2\cdot 4^n\sum_{j=0}^{n-1}c_{j}^14^{2j}+3\cdot\sum_{j=0}^{n-1}4^{2j+1}\pmod{4^{2n}-1}\notag\\
			\equiv&\left\{ \begin{array}{ll}
				2\sum\limits_{j=0}\limits^{n-1}(c_j^0-c_{j}^1)4^{2j}+3\cdot\sum\limits_{j=0}\limits^{n-1}4^{2j+1}, &\pmod{4^{n}-1}\\
				2\sum\limits_{j=0}\limits^{n-1}(c_j^0+c_{j}^1)4^{2j}+3\cdot\sum\limits_{j=0}\limits^{n-1}4^{2j+1}, &\pmod{4^{n}+1}
			\end{array} \right.\notag\\
			\equiv&\left\{ \begin{array}{ll}
				2\varepsilon_1\cdot\frac{4^{2p(p+2)}-1}{4^{2(p+2)}-1}, &\pmod{4^{n}-1}\\
				2\sum\limits_{j=0}\limits^{n-1}(c_j^0+c_{j}^1)4^{2j}+3\cdot\sum\limits_{j=0}\limits^{n-1}4^{2j+1}, &\pmod{4^{n}+1}
			\end{array} \right..
		\end{align}
		
		From $\gcd(4^n-1,4^n+1)=1$, we have $d=\gcd(S_1(4),4^{2n}-1)=\gcd(S_1(4),4^{n}-1)\cdot\gcd(S_1(4),4^{n}+1)=\frac{4^{p(p+2)}-1}{4^{p+2}-1}\cdot\gcd(S_1(4),4^{n}+1)$.
		Let $d_+=\gcd(S_1(4),4^{n}+1)$. In the following part we will determine $d_+$.
		
		Let $s_j^0=(-1)^{c_j^0}$, $s_j^1=(-1)^{c_j^1}$, that is $c_j^0=\frac{1}{2}(1-s_j^0)$ and $c_j^1=\frac{1}{2}(1-s_j^1)$,  then  we obtain
		\begin{align}\label{ES1.3}
			S_1(4)\equiv&	2\sum\limits_{j=0}\limits^{n-1}(c_j^0+c_{j}^1)4^{2j}+3\cdot\sum\limits_{j=0}\limits^{n-1}4^{2j+1}\pmod{4^{n}+1}\notag\\
			\equiv&	\sum\limits_{j=0}\limits^{n-1}(2-s_j^0-s_{j}^1)4^{2j}+3\cdot\sum\limits_{j=0}\limits^{n-1}4^{2j+1}\pmod{4^{n}+1}\\
			\equiv&	-\sum\limits_{j=0}\limits^{n-1}(s_j^0+s_{j}^1)4^{2j}\pmod{\frac{4^{n}+1}{5}}\notag.
		\end{align}
	
		Let $\pi$ be a prime divisor of $\gcd(\sum\limits_{j=0}\limits^{n-1}(s_j^0+s_{j}^1)4^{2j},\frac{4^{p(p+2)}+1}{5})$.
		From $\sum\limits_{j=0}\limits^{n-1}(s_j^0+s_{j}^1)4^{2j}\equiv0\pmod{\pi}$, we know that
		\begin{align}\label{E4}
			0&\equiv(\sum\limits_{i=0}^{n-1}(s_i^0+s_{i}^1)4^{2i})(\sum\limits_{j=0}^{n-1}(s_j^0+s_{j}^1)4^{-2j})\equiv\sum\limits_{i,j=0}^{n-1}(s_i^0+s_{i}^1)(s_j^0+s_{j}^1)4^{2(i-j)}\pmod{\pi}\notag\\
			&\equiv\sum\limits_{i,j=0}^{n-1}s_i^0s_j^04^{2(i-j)}+\sum\limits_{i,j=0}^{n-1}s_i^1s_j^14^{2(i-j)}+\sum\limits_{i,j=0}^{n-1}s_i^0s_j^14^{2(i-j)}+\sum\limits_{i,j=0}^{n-1}s_i^1s_j^04^{2(i-j)}\pmod{\pi}\notag\\&\equiv\sum\limits_{\tau,j=0}^{n-1}s_{j+\tau}^0s_j^04^{2\tau}+\sum\limits_{\tau,j=0}^{n-1}s_{j+\tau}^1s_j^14^{2\tau}+\sum\limits_{\tau,j=0}^{n-1}s_{j+\tau}^0s_j^14^{2\tau}+\sum\limits_{\tau,j=0}^{n-1}s_{j+\tau}^1s_j^04^{2\tau}\pmod{\pi}\notag\\&\equiv\sum\limits_{\tau=0}^{n-1}4^{2\tau}\sum\limits_{j=0}^{n-1}s_{j+\tau}^0s_{j}^0+\sum\limits_{\tau=0}^{n-1}4^{2\tau}\sum\limits_{j=0}^{n-1}s_{j+\tau}^1s_{j}^1+\sum\limits_{\tau=0}^{n-1}4^{2\tau}\sum\limits_{j=0}^{n-1}s_{j+\tau}^0s_{j}^1+\sum\limits_{\tau=0}^{n-1}4^{2\tau}\sum\limits_{j=0}^{n-1}s_{j+\tau}^1s_{j}^0\pmod{\pi}\notag\\
			&\equiv\sum\limits_{\tau=0}^{n-1}4^{2\tau}(R_{c^0}(\tau)+R_{c^1}(\tau)+2R_{c^0,c^1}(\tau))\pmod{\pi}\notag\\
			&\equiv(2p(p+2)+2p^2)+\sum\limits_{\substack{\tau=1\\p+2\mid\tau}}^{n-1}(-1-1+2(-2p-1))4^{2\tau}+\sum\limits_{\substack{\tau=1\\p+2\nmid\tau}}^{n-1}(-1+3+2\cdot1)4^{2\tau}\pmod{\pi}\notag\\
			&~~~~~~~~~~~~~~~~~~~~~~~~~~~~~~~~~~~~~~~~~~~~~~~~~~~~~~~~~~~~~~~~\text{(from equation (\ref{ideal}) and Lamme \ref{lem2})}\notag\\
			&\equiv4p(p+1)-4(p+1)\sum\limits_{\substack{\tau=1\\p+2\mid\tau}}^{n-1}4^{2\tau}+4\sum\limits_{\substack{\tau=1\\p+2\nmid\tau}}^{n-1}4^{2\tau}\pmod{\pi}\notag\\&\equiv4p(p+1)-4(p+2)\sum\limits_{\substack{\tau=1\\p+2\mid\tau}}^{n-1}4^{2\tau}+4\sum\limits_{\substack{\tau=1}}^{n-1}4^{2\tau}\pmod{\pi}\notag\\&\equiv4p(p+1)-4(p+2)\sum\limits_{\substack{i=1}}^{p-1}4^{2i(p+2)}-4\pmod{\pi}\notag\\&\equiv4p(p+1)-4(p+2)(\sum\limits_{\substack{i=0}}^{p-1}4^{2i(p+2)}-1)-4\pmod{\pi}\notag\\&\equiv4(p+1)^2-4(p+2)\cdot\frac{4^{p(p+2)}-1}{4^{p+2}-1}\cdot\frac{4^{p(p+2)}+1}{4^{p+2}+1}\pmod{\pi}.
		\end{align}
		If $\pi\mid\frac{4^{p(p+2)}+1}{4^{p+2}+1}$, then $4(p+1)^2\equiv0\pmod{\pi}$, $\pi\mid p+1$ and $\pi\mid4^{p(p+2)}+1$. It then follows that the order of 4 mod $\pi$ is $2p(p+2)$. From Little Fermat's Theorem we have $2p(p+2)\mid\pi-1$ is a contradiction. If $\pi\mid\frac{4^{p+2}+1}{5}$, then we have $$0\equiv4p(p+1)-4(p+2)\sum\limits_{\substack{i=1}}^{p-1}4^{2i(p+2)}-4\equiv4p^2+4p-4(p^2+p-2)-4=4\pmod{\pi},$$which is a contradiction. Hence, $\gcd(S(4),\frac{4^{n}+1}{5})=1$.
		From (\ref{ES1.2}) we know that
		\begin{align*}
			S_1(4)&\equiv	2\sum\limits_{j=0}\limits^{n-1}(c_j^0+c_{j}^1)4^{2j}+3\cdot\sum\limits_{j=0}\limits^{n-1}4^{2j+1}\pmod5\\
			&\equiv	2\sum\limits_{j=0}\limits^{n-1}(c_j^0+c_{j}^1)+3\cdot4\cdot\sum\limits_{j=0}\limits^{n-1}1\pmod5\\
			&\equiv2(\sum\limits_{\substack{j=0\\p+2\mid j}}\limits^{n-1}(0+1)+\sum\limits_{\substack{j=1\\p\mid j}}\limits^{n-1}(1+1)+2\cdot\frac{p(p+2)-p-(p+1)}{2})+2p(p+2)\pmod5\\
			&\equiv2(p+2(p+1)+p(p+2)-p-(p+1))+2p(p+2)\pmod5\\
			&\equiv4p^2+2\not\equiv5\pmod5\quad(\text{since } p^2\equiv 1, 5, 9\pmod{10}).
		\end{align*}
		To sum up, for $e=(1,\ 0,\ 0)$, $d_+=\gcd(S_1(4),4^{n}+1)=1$, $d=\gcd(S_1(4),4^{2n}-1)=\frac{4^{p(p+2)}-1}{4^{p+2}-1}$. The 4-adic complexity of sequence $s$ is $C_4(s)=\log_{4}(4^{p(p+2)}+1)(4^{p+2}-1)$.
		
		The other three cases can be proved similarly.
	\end{proof}
	
		\begin{theorem}\label{th2}
		Let $t_0$ and $t_1$ be the GMW sequences pair of length $n=2^{2k}-1$,  and $s=s(a,b)$ be the optimal quaternary sequence given by Lemma \ref{lem1}, then for $c^0=c^2, c^1=c^3$, the 4-adic complexity of the sequence $s$ is
		\begin{align*}
			C_{4}(s)
			=\left\{ \begin{array}{ll}
				\log_{4}(4^{2^{2k}-1}+1)(4^{2^k+1}-1),\ & \text{if}\ e=(1,0,0)\ \text{or} \ e=(1,1,1)\\
				\log_{4}(4^{2^{2k}-1}-1)(4^{2^k+1}+1),\ &\text{if}\ e=(0,1,0)\ \text{or} \ e=(0,0,1)\
			\end{array} \right..
		\end{align*}
	\end{theorem}
	\begin{proof}
		Let $c^0$ and $c^1$ be the GMW sequences pair of length $n=2^{2k}-1$. Then we have
		$$\sum_{j=0}^{n-1}(c_j^0- c_{j}^1)4^{2j}=\sum_{\substack{j=0\\j\equiv0\pmod{2^k+1}}}^{n-1}\varepsilon_2\cdot4^{2j}=\varepsilon_2\cdot\sum_{j=0}^{2^k-1-1}4^{2(2^k+1)j}=\varepsilon_2\cdot\frac{4^{2n}-1}{4^{2(2^k+1)}-1},$$
			where \begin{align*}
			\varepsilon_2=	\left\{ \begin{array}{ll}
				-1,\ &\text{if  $c^0$ is the GMW sequence and $c^1$ is the modified GMW sequence}\\1, \ &\text{if  $c^0$ is the modified GMW sequence and $c^1$ is the  GMW sequence}
			\end{array}\right..
		\end{align*}
	For $e=(1,0,0)$, from (\ref{ES1}) we obtain
		\begin{align}\label{ES1.4}
			S_2(4)=&2\sum_{j=0}^{n-1}c_j^04^{2j}-2\cdot 4^n\sum_{j=0}^{n-1}c_{j}^14^{2j}+3\cdot\sum_{j=0}^{n-1}4^{2j+1}\pmod{4^{2n}-1}\notag\\
			\equiv&\left\{ \begin{array}{ll}
				\varepsilon_2\cdot\frac{4^{2n}-1}{4^{2(2^k+1)}-1}, &\pmod{4^{n}-1}\\
				2\sum\limits_{j=0}\limits^{n-1}(c_j^0+c_{j}^1)4^{2j}+3\cdot\sum\limits_{j=0}\limits^{n-1}4^{2j+1}, &\pmod{4^{n}+1}
			\end{array} \right..
		\end{align}
		
		From $\gcd(4^n-1,4^n+1)=1$, we have $d=\gcd(S_2(4),4^{2n}-1)=\gcd(S_2(4),4^{n}-1)\cdot\gcd(S_1(4),4^{n}+1)=\frac{4^{2^{2k}-1}-1}{4^{2^k+1}-1}\cdot\gcd(S_2(4),4^{n}+1)$.
		Let $d_+=\gcd(S_2(4),4^{n}+1)$. In the following part we will determine $d_+$.
		
		Let $s_j^0=(-1)^{c_j^0}$, $s_j^1=(-1)^{c_j^1}$, which implies  $c_j^0=\frac{1}{2}(1-s_j^0)$ and $c_j^1=\frac{1}{2}(1-s_j^1)$, then from (\ref{ES1.4}) we have
		\begin{align}
			S_2(4)\equiv	-\sum\limits_{j=0}\limits^{n-1}(s_j^0+s_{j}^1)4^{2j}\pmod{\frac{4^{n}+1}{5}}\notag.
		\end{align}
		
		Let $\pi$ be a prime divisor of $\gcd(\sum\limits_{j=0}\limits^{n-1}(s_j^0+s_{j}^1)4^{2j},\frac{4^{2^{2k}-1}+1}{5})$.
		From $\sum\limits_{j=0}\limits^{n-1}(s_j^0+s_{j}^1)4^{2j}\equiv0\pmod{\pi}$, we know that
		\begin{align}\label{E4}
			0&\equiv(\sum\limits_{i=0}^{n-1}(s_i^0+s_{i}^1)4^{2i})(\sum\limits_{j=0}^{n-1}(s_j^0+s_{j}^1)4^{-2j})\equiv\sum\limits_{i,j=0}^{n-1}(s_i^0+s_{i}^1)(s_j^0+s_{j}^1)4^{2(i-j)}\pmod{\pi}\notag\\
		&\equiv\sum\limits_{\tau=0}^{n-1}4^{2\tau}(R_{c^0}(\tau)+R_{c^1}(\tau)+2R_{c^0,c^1}(\tau))\pmod{\pi}\notag\\
			&\equiv(2(2^{2k}-1)+2(2^{2k}-2^{k+1}+1))+\sum\limits_{\substack{\tau=1\\2^k+1\mid\tau}}^{n-1}(-1-1+2(-2^{k+1}+1))4^{2\tau}+\sum\limits_{\substack{\tau=1\\2^k+1\nmid\tau}}^{n-1}(-1+3+2\cdot1)4^{2\tau}\pmod{\pi}\notag\\&\equiv2(2^{2k+1}-2^{k+1})-4\cdot2^k\sum\limits_{\substack{\tau=1\\2^k+1\mid\tau}}^{n-1}4^{2\tau}+4\sum\limits_{\substack{\tau=1\\2^k+1\nmid\tau}}^{n-1}4^{2\tau}\pmod{\pi}\notag\\&\equiv4\cdot2^k(2^k-1)-4(2^k+1)\sum\limits_{\substack{\tau=1\\2^k+1\mid\tau}}^{n-1}4^{2\tau}+4\sum\limits_{\substack{\tau=1}}^{n-1}4^{2\tau}\pmod{\pi}\notag\\&\equiv4\cdot2^k(2^k-1)-4(2^k+1)\sum\limits_{\substack{i=1}}^{2^{k}-2}4^{2i(2^k+1)}-4\pmod{\pi}\\
			&\equiv2^{2(k+1)}-4(2^k+1)\frac{4^{2^{2k}-1}-1}{4^{(2^k+1)}-1}\cdot\frac{4^{2^{2k}-1}+1}{4^{(2^k+1)}+1}\pmod{\pi}\notag.
		\end{align}
		Since $\pi$ is an odd prime, it is obvious that $\pi\nmid\frac{4^{2^{2k}-1}+1}{4^{(2^k+1)}+1}$. If $\pi\mid\frac{4^{2^{2k}-1}+1}{4^{(2^k+1)}+1}$, we have the contradiction that $0\equiv2^{2(k+1)}\pmod{\pi}$.  Therefore we have $\pi\mid\frac{4^{(2^k+1)}+1}{5}$,
		from (\ref{E4}) can obtain that
		$$0\equiv4\cdot2^k(2^k-1)-4(2^k+1)\sum\limits_{\substack{i=1}}^{2^{k}-2}1-4\equiv4\pmod{\pi}$$
		which is a contradiction. This implies $\gcd(S_2(4),\frac{4^{2^{2k}-1}+1}{5})$=1.
		
		And from (\ref{ES1.4}) we have
		\begin{align*}
			S_2(4)\equiv&
				2\sum\limits_{j=0}\limits^{n-1}(c_j^0+c_{j}^1)+2\cdot\sum\limits_{j=0}\limits^{n-1}1 \pmod{5}\\
				\equiv&-\sum\limits_{j=0}\limits^{n-1}(s_j^0+s_{j}^1)+4\cdot\sum\limits_{j=0}\limits^{n-1}1 \pmod{5}.
		\end{align*}
If $5\mid S_2(4) $,	then we obtain
	\begin{align*}
		n^2\equiv(4\cdot\sum\limits_{j=0}\limits^{n-1}1)^2&\equiv(\sum\limits_{j=0}\limits^{n-1}(s_j^0+s_{j}^1))^2 \pmod{5}\\
		&\equiv\sum\limits_{\tau=0}^{n-1}(R_{c^0}(\tau)+R_{c^1}(\tau)+2R_{c^0,c^1}(\tau))\pmod{5},
	\end{align*}
that is 
\begin{align*}
	(2^{2k}-1)^2\equiv&(2(2^{2k}-1)+2(2^{2k}-2^{k+1}+1))+\sum\limits_{\substack{\tau=1\\2^k+1\mid\tau}}^{n-1}(-1-1+2(-2^{k+1}+1))\\
	&+\sum\limits_{\substack{\tau=1\\2^k+1\nmid\tau}}^{n-1}(-1+3+2\cdot1)\pmod{5}\\
	\equiv&2^{2k+2}-2^{k+2}-2^{k+2}(2^k-2)+4((2^{2k}-2)-(2^k-2))\pmod{5}\\
		\equiv&4\cdot2^{2k}\equiv-2^{2k}\pmod{5}.
\end{align*}
This shows $0\equiv(2^{2k}-1)^2+2^{2k}\equiv4^{2k}-4^k+1\equiv(-1)^{2k}-(-1)^k+1\equiv2-(-1)^k\not\equiv0\pmod5$, which is a contradiction.

	In a conclusion, $d=\gcd(S_2(4),4^{2n}-1)=\gcd(S_2(4),4^{n}-1)\cdot\gcd(S_2(4),4^{n}+1)=\frac{4^{2^{2k}-1}-1}{4^{2^k+1}-1}$. The 4-adic complexity of sequence $s$ is $C_4(s)=	\log_{4}(4^{2^{2k}-1}+1)(4^{2^k+1}-1).$
			The remaining cases can be proved in the same way.
		\end{proof}
	
	\subsection{Condition \uppercase\expandafter{\romannumeral2}: With the notations as before, for $t_0$ and $t_1$ are the twin-prime sequences pair of length $p(p+2)$ or GMW sequences pair of length $2^{2k}-1$, and  $c^0=c^3, c^1=c^2$.}
For $c^0=c^3, c^1=c^2$, we have $a_{2j}=c_j^0$, $a_{2j+1}\equiv c_{j+\lambda}^1+e_0\pmod2$ , $b_{2j}\equiv c_j^1+e_1\pmod 2,$ and $ b_{2j+1}\equiv c_{j+\lambda}^0+e_2\pmod2$.
 
From the definition of  $s$, $S(N)$ and $\lambda=\frac{n+1}{2}$, we know that
\begin{align}\label{ES2}
	S(4)\notag
	=&\sum_{i=0}^{N-1}s_i4^i\pmod{4^{2n}-1}\notag\\
    =&2(\sum_{\substack{j=0\\a_{2j}=1}}^{n-1}4^{2j}+\sum_{\substack{j=0\\a_{2j+1}=1}}^{n-1}4^{2j+1})+
	(\sum_{\substack{j=0\\a_{2j}+b_{2j}=1}}^{n-1}4^{2j}+\sum_{\substack{j=0\\a_{2j+1}+b_{2j+1}=1}}^{n-1}4^{2j+1})\pmod{4^{2n}-1}\notag\\
	=&2\cdot\sum_{\substack{j=0\\c_j^0=1}}^{n-1}4^{2j}+2\cdot\sum_{\substack{j=0\\c_{j+\lambda}^1+e_0\equiv1\pmod2}}^{n-1}4^{2j+1}+
	\sum_{\substack{j=0\\c_j^0+c_j^1+e_1\equiv1\pmod2}}^{n-1}4^{2j}\notag\\&+\sum_{\substack{j=0\\c_{j+\lambda}^1+c_{j+\lambda}^0+e_0+e_2\equiv1\pmod2}}^{n-1}4^{2j+1}\pmod{4^{2n}-1}\notag\\
	=&2\cdot\sum_{\substack{j=0\\c_j^0=1}}^{n-1}4^{2j}+2\cdot4^n\cdot\sum_{\substack{j=0\\c_{j}^1+e_0\equiv1\pmod2}}^{n-1}4^{2j}+
	\sum_{\substack{j=0\\c_j^0+c_j^1+e_1\equiv1\pmod2}}^{n-1}4^{2j}\notag\\&+4^n\cdot\sum_{\substack{j=0\\c_{j}^1+c_{j}^0+e_0+e_2\equiv1\pmod2}}^{n-1}4^{2j}\pmod{4^{2n}-1}.
\end{align}
Let	$c^0$ and $c^1$ be the twin-prime sequences pair with period $p(p+2)$, we have
\begin{align*}
	c^0_i+c^1_i\equiv\left\{ \begin{array}{ll}
		0\pmod2,&\ i\not\equiv0\pmod{p+2}\\
		1\pmod2,&\ i\equiv0\pmod{p+2}
	\end{array} \right.,
\end{align*}
and let $c^0$ and $c^1$ be the GMW sequences pair  with period  $2^{2k}-1$, we have
\begin{align*}
	c^0_i+c^1_i\equiv\left\{ \begin{array}{ll}
		0\pmod2,&\ i\not\equiv0\pmod{2^k+1}\\
		1\pmod2,&\ i\equiv0\pmod{2^k+1}
	\end{array} \right..
\end{align*}	Assume that $c$ and $c'$ is the twin-prime sequence and modified twin-prime sequence with period  $n=p(p+2)$, respectively, then we obtain
$$\sum_{j=0}^{n-1}(c_j^0+c_{j}^1)4^{2j}+\sum_{\substack{j=0\\j\equiv0\pmod{p+2}}}^{n-1}4^{2j}=2\sum_{j=0}^{n-1}c'_j4^{2j}$$and
\begin{align}\label{p}
\sum_{j=0}^{n-1}(c_j^0-c_{j}^1)4^{2j}=\sum_{\substack{j=0\\j\equiv0\pmod{p+2}}}^{n-1}\varepsilon_1\cdot4^{2j}=\varepsilon_1\sum_{j=0}^{p-1}4^{2(p+2)j}=\varepsilon_1\frac{4^{n}-1}{4^{2(p+2)}-1},
\end{align}
	where \begin{align*}
	\varepsilon_1=	\left\{ \begin{array}{ll}
		-1,\ &\text{if  $c^0$ is the twin-prime sequence and $c^1$ is the modified twin-prime sequence}\\1, \ &\text{if  $c^0$ is the modified twin-prime sequence and $c^1$ is the  twin-prime sequence}
	\end{array}\right..
\end{align*}
Let $e=(1,0,0)$, and 	$c^0$, $c^1$ be the twin-prime sequences pair with period $p(p+2)$, from equation (\ref{ES2}) we know that
\begin{align}\label{ES2.2}
S_2(4)\notag
=&2\cdot\sum_{\substack{j=0\\c_j^0=1}}^{n-1}4^{2j}+2\cdot4^n\cdot\sum_{\substack{j=0\\c_{j}^1+e_0\equiv1\pmod2}}^{n-1}4^{2j}+
\sum_{\substack{j=0\\c_j^0+c_j^1+e_1\equiv1\pmod2}}^{n-1}4^{2j}\notag\\&+4^n\cdot\sum_{\substack{j=0\\c_{j}^1+c_{j}^0+e_0+e_2\equiv1\pmod2}}^{n-1}4^{2j}\pmod{4^{2n}-1}\notag\\
=&2\cdot\sum_{\substack{j=0\\c_j^0=1}}^{n-1}4^{2j}+2\cdot4^n\cdot\sum_{\substack{j=0\\c_{j}^1=0}}^{n-1}4^{2j}+
\sum_{\substack{j=0\\c_j^0+c_j^1=1}}^{n-1}4^{2j}+4^n\cdot\sum_{\substack{j=0\\c_{j}^1+c_{j}^0=0\pmod2}}^{n-1}4^{2j}\pmod{4^{2n}-1}\notag\\
=&2\sum_{j=0}^{n-1}c_j^04^{2j}+2\cdot4^n\cdot\sum_{j=0}^{n-1}(1-c_{j}^1)4^{2j}+
\sum_{\substack{j=0\\j\equiv0\pmod{p+2}}}^{n-1}4^{2j}+4^n\cdot\sum_{\substack{j=0\\j\not\equiv0\pmod{p+2}}}^{n-1}4^{2j}\pmod{4^{2n}-1}\notag\\
=&2\sum_{j=0}^{n-1}c_j^04^{2j}-2\cdot 4^n\sum_{j=0}^{n-1}c_{j}^14^{2j}+3\cdot4^n\cdot\sum_{j=0}^{n-1}4^{2j}-(4^n-1)\cdot\sum_{\substack{j=0\\j\equiv0\pmod{p+2}}}^{n-1}4^{2j}\pmod{4^{2n}-1}\notag\\
=&\left\{ \begin{array}{ll}
	2\sum\limits_{j=0}\limits^{n-1}c_j^04^{2j}-2\cdot\sum\limits_{j=0}\limits^{n-1}c_{j}^14^{2j}+3\cdot\sum\limits_{j=0}\limits^{n-1}4^{2j},&\ \pmod{4^n-1}\\
	2\sum\limits_{j=0}\limits^{n-1}c_j^04^{2j}+2\cdot\sum\limits_{j=0}\limits^{n-1}c_{j}^14^{2j}-3\cdot\sum\limits_{j=0}\limits^{n-1}4^{2j}+2\cdot\sum\limits_{\substack{j=0\\j\equiv0\pmod{p+2}}}\limits^{n-1}4^{2j},&\ \pmod{4^n+1}
\end{array} \right.\notag\\
\equiv&\left\{ \begin{array}{ll}
	2\varepsilon_1\cdot\sum\limits_{j=0}\limits^{p-1}4^{2(p+2)j}, &\pmod{4^{n}-1}\\
	4\sum\limits_{j=0}\limits^{n-1}c'_j4^{2j}+2\cdot\frac{4^n+1}{5},& \pmod{4^{n}+1}\ (\text{since}\ 4^{n}\equiv 4\pmod5)
\end{array} \right..
\end{align}

Similarly, for $e=(0,1,0)$, we have
\begin{align*}
	S(4)
	\equiv\left\{ \begin{array}{ll}
		2\sum\limits_{j=0}\limits^{n-1}c_j^04^{2j}+2\cdot\sum\limits_{j=0}\limits^{n-1}c_{j}^14^{2j}, &\pmod{4^{n}-1}\\
		(2+2\varepsilon_1)\sum\limits_{j=0}\limits^{p-1}4^{2(p+2)j}+2\cdot\frac{4^n+1}{5},& \pmod{4^{n}+1}
	\end{array} \right..
\end{align*}

For $e=(0,0,1)$,  we have
\begin{align*}
	S(4)
	\equiv\left\{ \begin{array}{ll}
		2\sum\limits_{j=0}\limits^{n-1}c_j^04^{2j}+2\cdot\sum\limits_{j=0}\limits^{n-1}c_{j}^14^{2j}+\frac{4^n-1}{3}, &\pmod{4^{n}-1}\\
		(2+2\varepsilon_1)\sum\limits_{j=0}\limits^{p-1}4^{2(p+2)j}-\frac{4^n+1}{5},& \pmod{4^{n}+1}
	\end{array} \right..
\end{align*}

For $e=(1,1,1)$, we obation
\begin{align*}
	S(4)
	\equiv\left\{ \begin{array}{ll}
		2\varepsilon_1\cdot\sum\limits_{j=0}\limits^{p-1}4^{2(p+2)j}, &\pmod{4^{n}-1}\\
		4\sum\limits_{j=0}\limits^{n-1}c'_j4^{2j}+2\cdot\frac{4^n+1}{5},& \pmod{4^{n}+1}
	\end{array} \right..
\end{align*}

Let $c^0$ and $c^1$ be the GMW sequences pair of length $n=2^{2k}-1$, we can obtain $S(4)$ in the same way. And simiar to the proof of Theorem \ref{th1} and \ref{th2}, we came up with the following results.
		\begin{theorem}\label{th3}
		Let $t_0$ and $t_1$ be the twin-prime sequences pair with period $n=p(p+2)$, and $s=s(a,b)$ be the optimal quaternary sequence given by Lemma \ref{lem1}, for $c^0=c^3, c^1=c^2$, then the 4-adic complexity of the sequence $s$ is
		\begin{align*}
			C_{4}(s)
			=\left\{ \begin{array}{ll}
				\log_{4}(4^{p(p+2)}+1)(4^{p+2}-1),\ & \text{if}\ e=(1,0,0)\ \text{or} \ e=(1,1,1)\\
				\log_{4}5(4^{p(p+2)}-1),\ &\text{if}\ c^0=t_0,\ c^1=t_1,\ e=(0,1,0)\ \text{or} \ e=(0,0,1)\\
				\log_{4}(4^{p(p+2)}-1)(4^{p+2}+1),\ &\text{if}\  c^0=t_1,\ c^1=t_0,\ e=(0,1,0)\ \text{or} \ e=(0,0,1)
			\end{array} \right..
		\end{align*}
	\end{theorem}
	\begin{theorem}\label{th4}
	Let $t_0$ and $t_1$ be the GMW sequences pair of length $n=2^{2k}-1$, and $s=s(a,b)$ be the optimal quaternary sequence given by Lemma \ref{lem1}, for $c^0=c^3, c^1=c^2$, then the 4-adic complexity of the sequence $s$ is
	\begin{align*}
		C_{4}(s)
		=\left\{ \begin{array}{ll}
			\log_{4}(4^{2^{2k}-1}+1)(4^{2^k+1}-1),\ & \text{if}\ e=(1,0,0)\ \text{or} \ e=(1,1,1)\\
			\log_{4}5(4^{2^{2k}-1}-1) ,\ &\text{if}\ c^0=t_0,\ c^1=t_1,\ e=(0,1,0)\ \text{or} \ e=(0,0,1)\\			\log_{4}(4^{2^{2k}-1}-1)(4^{2^k+1}+1),\ &\text{if}\ c^0=t_1,\ c^1=t_0,\ e=(0,1,0)\ \text{or} \ e=(0,0,1)
		\end{array} \right..
	\end{align*}
\end{theorem}
	\begin{example}
		Let $p=3$. Then the pair of twin-prime sequences with period $n=p(p+2)$ are given by
		\begin{align*}
			t_0=(000100110101111),
		\end{align*}
		\begin{align*}
			t_1=(100101110111111).
		\end{align*}
		For $(c^0,c^1,c^2,c^3)=(t_0,t_1,t_0,t_1)$ and $(e_0,e_1,e_2)=(1,0,0)$, we have 
		\begin{align*}
			a=(010000100000101001110011101010),
		\end{align*}
		\begin{align*}   
			b=(000101110101111100100110111111),
		\end{align*}
		and 
		\begin{align*}
			s=(030101210101212103230123212121).
		\end{align*}
		Computing by Magma program, we obtain $\gcd(S(4),4^{2n}-1)=1 049 601=\frac{4^{15}-1}{4^5-1}$. Then the 4-adic complexity of the sequence $s$ is $C_{4}(s)=\log_{4}(4^{15}+1)(4^5-1)$, which is consistent with Theorem \ref{th1}.
	\end{example}
	
		\begin{example}
		
		Let k=3. The GMW sequences pair with period $n=2^{2k}-1=63$ are given by
		\begin{align*}
			t_0=(000001000011000101001111010001110010010110111011001101010111111),
		\end{align*}
		\begin{align*}
			t_1=(100001000111000101101111010101110010110110111111001101110111111).
		\end{align*}
		
		For $(c^0,c^1,c^2,c^3)=(t_0,t_1,t_0,t_1)$ and $(e_0,e_1,e_2)=(0,1,0)$, we have 
		\begin{align*}
			s=(010103010323010303012323030303230121030321232323012303030323232\\30
			1012101032101212303232121012321030321230323032103212321232323).
		\end{align*}
		Computing by Magma program, we obtain $\gcd(S(4),4^{2n}-1)=324517315723109789871420976398337=\frac{4^{63}+1}{4^9+1}$. Then the 4-adic complexity of the sequence $s$ is $C_{s}(4)=\log_{4}(4^{63}-1)(4^9+1)$,  which is consistent with Theorem \ref{th2}.
	\end{example}

		\section{The 4-Adic Complexity of optimal autocorrelation quaternary sequence $s$ constructed by binary cyclotomic sequences of order four}\label{sec4}
	In this section, we will determind  the 4-adic complexity of the sequences $s$ with period $2n$ ($n=4f+1=x^2+4y^2$ is prime) given by Lemma \ref{lem7}-\ref{lem9}. Fristly, we give the lemma about the values of the cyclotomic numbers of order four in $\mathbb{F}_n$.
		\begin{lemma}\label{lem4.1}(\cite{B. C. Barndt},\cite{T. Storer})
			Let $n=4f+1=x^2+4y^2$ be an odd prime,  x, y$\in\mathbb{Z}$ and f be an odd number. The values of the cyclotomic numbers $(i, j)$ of order four in $\mathbb{F}_n$ are
			$$16(0,0)=16(2,0)=16(2,2)=A=n-7+2x,$$
			$$16(0,1)=16(1,3)=16(3,2)=B=n+1+2x-8y,$$
			$$16(0,3)=16(1,2)=16(3,1)=\bar{B}=n+1+2x+8y,$$
			$$16(0,2)=C=n+1-6x,$$
			$$16(1,0)=16(1,1)=16(2,1)=16(2,3)=16(3,0)=16(3,3)=D=n-3-2x.$$
		\end{lemma}
		
		
	In the following we give a brief introduction about ``Gauss periods" of order four and ``quadratic Gauss sums".
		
		Since $\alpha\equiv \beta\pmod {2n}$ which implies $4^{2\alpha}\equiv 4^{2\beta} \bmod (4^{2n}-1) ,$ we can define the following mapping
		$$f: \mathbb{Z}_n\rightarrow \mathbb{Z}_{4^{2n}-1}^{\ast}, f(\alpha)= 4^{2\alpha}$$
		where $\mathbb{Z}_m^{\ast}$ is the group of units in the ring $\mathbb{Z}_m=\mathbb{Z}/m\mathbb{Z} \ (m\geq 2)$. $f$ is a homomorphism of groups from $(\mathbb{Z}_n, +)$ to $(\mathbb{Z}_{4^{2n}-1}^{\ast}, \cdot)$, and can be viewed as an additive character of finite field $\mathbb{Z}_n=\mathbb{F}_n$ valued in $\mathbb{Z}_{4^{2n}-1}^{\ast}$. Then we have the ``Gauss periods" of order 4
		$$\eta_\gamma=\sum_{i\in D_\gamma}4^{2i}\bmod (4^{2n}-1) (\gamma=0, 1, 2, 3),$$
		where $D_\gamma=\theta^\gamma C \ (0\leq\gamma\leq3)$ is the cyclotomic classes of order four in $\mathbb{F}_n$.
		and ``quadratic Gauss sums"
		$$G=\sum_{i\in \mathbb{F}_n^{\ast}}4^{2i}\chi(i)=\eta_0-\eta_1+\eta_2-\eta_3\pmod{4^{2n}-1},$$
		where $\chi$ is the quadratic (multiplicative) character of $\mathbb{F}_n^{\ast}$ (the Legendre symbol). Namely, for $i\in \mathbb{F}_n^{\ast}$,
		
		\begin{align*}
			\chi(i)
			= \left\{ \begin{array}{ll}
				1, & \textrm{if $i\in D_0\cup D_2=\langle\theta^2\rangle$}\\
				-1, & \textrm{if $i\in D_1\cup D_3=\theta\langle\theta^2\rangle$}
			\end{array} \right..
		\end{align*}
		
		The following results show that $\eta_\gamma$ and $G$ have some similar properties as usual Gauss periods and Gauss sums. And the  proof of the following lemmas are similar to \cite{Zhang} and \cite{Jing}, we omit them here.
		\begin{lemma}\label{lem4.3}
			(1). $G^2\equiv n-\frac{4^{2n}-1}{15}\ (\bmod \ 4^{2n}-1)$.\\
			(2). For $0\leq\gamma,\ \mu\leq3,$
			$\eta_\gamma\eta_\mu\equiv\frac{n-1}{4}\delta_{\gamma, \mu+2}+\sum\limits_{\nu=0}\limits^3(\gamma-\nu+2, \mu-\nu)\eta_\nu \ (\bmod \ 4^{2n}-1)$,
			where $(i, j)$ is the cyclotomic number of order four on $\mathbb{F}_n$ and
			\begin{align*}
				\delta_{\gamma, \mu}
				= \left\{ \begin{array}{ll}
					1, & \textrm{if $\gamma\equiv\mu \pmod 4$}\\
					0, & \textrm{otherwise}
				\end{array} \right..
			\end{align*}
		\end{lemma}
		
		\begin{lemma}\label{lem4.4}
			Let $n=4f+1=x^2+4y^2$ be an odd prime, $f$ be odd and $y=-1$.  Then
			\begin{align*}
				16\eta_\gamma^2& =A\eta_\gamma+B\eta_{\gamma+1}+C\eta_{\gamma+2}+\overline{B}\eta_{\gamma+3}\\
				&= n(\eta_{\gamma+0}+\eta_{\gamma+1}+\eta_{\gamma+2}+\eta_{\gamma+3})+(-7+2x)\eta_\gamma+(9+2x)\eta_{\gamma+1} +(1-6x)\eta_{\gamma+2}+(-7+2x)\eta_{\gamma+3}\notag\\
				&=(\frac{4^{2n}-1}{15}-1)n+(-7+2x)(\eta_\gamma+\eta_{\gamma+3})+(9+2x)\eta_{\gamma+1} +(1-6x)\eta_{\gamma+2},
			\end{align*}
			\begin{equation*}
				16\eta_{\gamma}\eta_{\gamma+1}=(\frac{4^{2n}-1}{15}-1)n+(-3-2x)(\eta_\gamma+\eta_{\gamma+1})+(-7+2x)\eta_{\gamma+2}+(9+2x)\eta_{\gamma+3},
			\end{equation*}
			\begin{equation*}
				16\eta_{\gamma}\eta_{\gamma+2}=(\frac{4^{2n}-1}{15}-1)n+(-7+2x)(\eta_\gamma+\eta_{\gamma+2})+(-3-2x)(\eta_{\gamma+1}+\eta_{\gamma+3})+4(n-1).
			\end{equation*}
		\end{lemma}
		\begin{lemma}\label{lem4.5} Let $n=4f+1=x^2+4y^2$ be an odd prime, $f$ be odd and $y=-1$.  Then
			$$(\eta_0-\eta_1)^2+ (\eta_2-\eta_3)^2=G,\ (\eta_0-\eta_3)^2+ (\eta_2-\eta_1)^2=-G.$$\ 
		\end{lemma}

	From construction (\ref{E1}) we know that
	$a_{2j}=c^0_j,\ a_{2j+1}=c^1_{j+\lambda}\oplus e_0,\ 
	b_{2j}=c^2_j\oplus e_1,\ b_{2j+1}=c^3_{j+\lambda}\oplus e_2,\ 0\leq j\leq n-1$, then from the definition of $S(N)$ we have
	\begin{align*}
		S(4)\notag
		=&\sum_{i=0}^{N-1}s_i4^i\pmod{4^{N}-1}\notag\\ =&\sum_{i=0}^{2n-1}\phi(a_i,b_i)4^i \pmod{4^{2n}-1}\notag\\
		=&\sum_{\substack{i=0\\a_{i}=0,b_{i}=1}}^{2n-1}4^i+\sum_{\substack{i=0\\a_{i}=1,b_{i}=1}}^{2n-1}2\cdot4^i+\sum_{\substack{i=0\\a_{i}=1, b_{i}=0}}^{2n-1}3\cdot4^i
		=2\sum_{\substack{i=0\\a_{i}=1}}^{2n-1}4^i+\sum_{\substack{i=0\\a_{i}+b_{i}=1}}^{2n-1}4^i\pmod{4^{2n}-1}\notag\\
		=&2(\sum_{\substack{j=0\\a_{2j}=1}}^{n-1}4^{2j}+\sum_{\substack{j=0\\a_{2j+1}=1}}^{n-1}4^{2j+1})+(\sum_{\substack{j=0\\a_{2j}+b_{2j}=1}}^{2n-1}4^{2j}+\sum_{\substack{j=0\\a_{2j+1}+b_{2j+1}=1}}^{2n-1}4^{2j+1})\pmod{4^{2n}-1}\notag\\
		=&2(\sum_{\substack{j=0\\c^0_j=1}}^{n-1}4^{2j}+\sum_{\substack{j=0\\c^1_{j+\lambda}\oplus e_0=1}}^{n-1}4^{2j+1})+
		(\sum_{\substack{j=0\\c^0_j+c^2_j\oplus e_1=1}}^{n-1}4^{2j}+\sum_{\substack{j=0\\c^1_{j+\lambda}\oplus e_0+c^3_{j+\lambda}\oplus e_2=1}}^{n-1}4^{2j+1})\pmod{4^{2n}-1}\notag.
	\end{align*}
	
	If $e=(0,0,0)$, we have
	\begin{align}\label{ES5}
		S(4)=&2(\sum_{\substack{j=0\\c^0_j=1}}^{n-1}4^{2j}+\sum_{\substack{j=0\\c^1_{j+\lambda}=1}}^{n-1}4^{2j+1})+
		(\sum_{\substack{j=0\\c^0_j+c^2_j=1}}^{n-1}4^{2j}+\sum_{\substack{j=0\\c^1_{j+\lambda}+c^3_{j+\lambda}=1}}^{n-1}4^{2j+1})\pmod{4^{2n}-1}\notag\\
		=&2\sum_{\substack{j=0\\c^0_j=1}}^{n-1}4^{2j}+2\cdot 4^n\cdot\sum_{\substack{j=0\\c^1_{j}=1}}^{n-1}4^{2j}+
		(\sum_{\substack{j=0\\c^0_j+c^2_j=1}}^{n-1}4^{2j}+4^n\cdot\sum_{\substack{j=0\\c^1_{j}+c^3_{j}=1}}^{n-1}4^{2j})\pmod{4^{2n}-1}.
	\end{align}
		Taking $(c^0,c^1,c^2,c^3)=(t_2,t_1,t_2,t_1)$ in Lemma \ref{lem7} as an example, then we obtain
		\begin{align}\label{ES3}
S(4)=&2\sum_{\substack{j=0\\t_2(j)=1}}^{n-1}4^{2j}+2\cdot 4^n\cdot\sum_{\substack{j=0\\t_1(j)=1}}^{n-1}4^{2j}\pmod{4^{2n}-1}\notag\\
			=&2\sum_{\substack{j=0\\j\in D_0\cup D_2}}^{n-1}4^{2j}+2\cdot 4^n\cdot\sum_{\substack{j=0\\j\in D_0\cup D_1}}^{n-1}4^{2j}\pmod{4^{2n}-1}\\
			=&2(\eta_0+\eta_2)+2\cdot 4^n(\eta_0+\eta_1)\pmod{4^{2n}-1}\notag\\
			=&\left\{ \begin{array}{ll}
				2(\eta_0+\eta_2+\eta_0+\eta_1)\pmod{4^{n}-1}\\
				2(\eta_0+\eta_2-\eta_0-\eta_1)\pmod{4^{n}+1}
			\end{array}\right.\notag\\
			=&\left\{ \begin{array}{ll}
				2(\eta_0+\sum\limits_{j=1}^{n-1}4^{2j}-\eta_3)\pmod{4^{n}-1}\\
				2(\eta_2-\eta_1)\pmod{4^{n}+1}
			\end{array}\right.\notag\\
			=&\left\{ \begin{array}{ll}
				2(\eta_0-\eta_3-1)\pmod{\frac{4^{n}-1}{3}}\\
				2(\eta_2-\eta_1)\pmod{\frac{4^{n}+1}{5}}
			\end{array}\right..\notag
		\end{align}
		From the above formula, we can see that to determine $\gcd(S(4), 4^{2n-1})$, we need to calculate $\gcd(\eta_{\gamma}-\eta_{\gamma-1}-1, \frac{4^{n}-1}{3})$ and $\gcd(\eta_{\gamma}-\eta_{\gamma-1}, \frac{4^{n}+1}{5})$ first.
		\begin{lemma}\label{lem5.1}
			Assume that $d_1$ is the maximum divisor of	$\gcd(\eta_{\gamma}-\eta_{\gamma-1}-1, \frac{4^{n}-1}{3})$, then we have
			$d_1\mid n^2+3n+4$. 
		\end{lemma}
	\begin{proof}
	For $\gamma=1$, on the one hand, from $d_1\mid\gcd(\eta_{1}-\eta_{0}-1,  \frac{4^{n}-1}{3})$, we have $$1\equiv\eta_{1}-\eta_{0}\pmod{d_1}.$$
		And on the other hand,
		$$G=\eta_{0}-\eta_{1}+\eta_{2}-\eta_{3}\equiv(\eta_{2}-\eta_{3})-1\pmod{d_1}.$$
		Combining with Lemmas \ref{lem4.5} and \ref{lem4.3} we obtain
		$$G\equiv(\eta_{1}-\eta_{0})^2+(\eta_{2}-\eta_{3})^2\equiv 1+(G+1)^2\equiv n+2+2G\pmod{d_1},$$
		which shows $-G\equiv n+2\pmod{d_1}$,  $n\equiv(-G)^2\equiv (n+2)^2\pmod{d_1}$, hence $d_1\mid n^2+3n+4$.  For the three other cases, we have the same result.
	\end{proof}
\begin{lemma}\label{lem5.2}
	Assume that $d_2$ is the maximum divisor of	$\gcd(\eta_{\gamma}-\eta_{\gamma-1}, \frac{4^{n}+1}{5})$, then we have
	$d_2=1$. 
\end{lemma}
\begin{proof}
	 For $\gamma=0$, we have $$0\equiv\eta_{0}-\eta_{3}\pmod{d_2}.$$
	$$G=\eta_{0}-\eta_{1}+\eta_{2}-\eta_{3}\equiv\eta_{2}-\eta_{1}\pmod{d_2}.$$
	From Lemmas \ref{lem4.5} and \ref{lem4.3} we obtain
	$$-G\equiv(\eta_{0}-\eta_{3})^2+(\eta_{2}-\eta_{1})^2\equiv G^2\equiv n\pmod{d_2},$$
	which implies  $n\equiv(-G)^2\equiv n^2\pmod{d_2}$, $d_2\mid n^2-n=n(n-1)$, then we have $d_2=n$ or $d_2\mid n-1$. Assume that $\pi$ is a prime divisor of $d_2$, since $d_2\mid\frac{4^{n}+1}{5}$,  we have the order of 4 mod $\pi$ is $2n$, and $2n\mid \pi-1$, which contradicts $\pi\mid n$ or $\pi\mid n-1$, therefore we arrive at $d_2=1$. For the three other cases, we have the same result.
\end{proof}
	
		\begin{theorem}\label{th4.1}
			Let $t_i\ (1\leq i\leq 6)$ be six binary sequences of length $n=4f+1$ with support sets $D_0\cap D_1$, $D_0\cap D_2$, $D_0\cap D_3$, $D_1\cap D_2$, $D_1\cap D_3$, $D_2\cap D_3$, respectively, and  $s=s(a,b)$ be the optimal quaternary sequence given by Lemma \ref{lem7}, then the 4-adic complexity of the sequence $s$ is
			\begin{align*}
				C_{4}(s)
				=\left\{ \begin{array}{ll}
					\log_{4}(\frac{4^{2n}-1}{15d_1}),\ & \text{if}\ e=(0,0,0),\ (1,0,1)\ \text{and}\ 3\mid f \ \text{or}\ e=(1,1,0),\ (0,1,1)\ \text{and}\ 3\mid f+1 \\
					\log_{4}(\frac{4^{2n}-1}{5d_1}),\ &\text{if}\ e=(0,0,0),\ (1,0,1)\ \text{and}\ 3\nmid f \ \text{or}\ e=(1,1,0),\ (0,1,1)\ \text{and}\ 3\nmid f+1
				\end{array} \right.,
			\end{align*}
		where $d_1\mid \gcd(n^2+3n+4,\ \frac{4^n-1}{3})$.
		\end{theorem}
		\begin{proof}
			From equation (\ref{ES3}) we know that
		\begin{align*}
			S(4)&\equiv2\sum_{\substack{j=0\\j\in D_0\cup D_2}}^{n-1}4^{2j}+2\cdot 4^n\cdot\sum_{\substack{j=0\\j\in D_0\cup D_1}}^{n-1}4^{2j}\pmod{4^{2n}-1}\\
			&=\left\{ \begin{array}{ll}
				2\sum\limits_{\substack{j=0\\j\in D_0\cup D_2}}^{n-1}1+2\cdot \sum\limits_{\substack{j=0\\j\in D_0\cup D_1}}^{n-1}1\pmod{3}\\
				2\sum\limits_{\substack{j=0\\j\in D_0\cup D_2}}^{n-1}1-2\cdot \sum\limits_{\substack{j=0\\j\in D_0\cup D_1}}^{n-1}1\pmod{5}
			\end{array} \right.\\
		&=\left\{ \begin{array}{ll}
			2\cdot 2f+2\cdot 2f\pmod{3}\\
				2\cdot 2f-2\cdot 2f\pmod{5}
		\end{array} \right.\\
	&=\left\{ \begin{array}{ll}
		2f\pmod{3}\\
		0\pmod{5}
	\end{array} \right..
		\end{align*}
	
	 Similarly, for $e=(1,0,1)$, we have
	\begin{align*}
		S(4)
		=&2\sum_{\substack{j=0\\j\in D_{i_1}\cup D_{i_2}}}^{n-1}4^{2j}+2\cdot 4^n\cdot\sum_{\substack{j=0\\j\in D_{i_3}\cup D_{i_4}}}^{n-1}4^{2j}\pmod{4^{2n}-1}\\
		&=\left\{ \begin{array}{ll}
			2f\pmod{3}\\
			0\pmod{5}
		\end{array} \right.,
	\end{align*}
	 and for $e=(1,1,0),\ (0,1,1)$, we obtain
	\begin{align*}
		S(4)
		=&2\sum_{\substack{j=0\\j\in D_{i_1}\cup D_{i_2}}}^{n-1}4^{2j}+2\cdot 4^n\cdot\sum_{\substack{j=0\\j\in D_{i_3}\cup D_{i_4}}}^{n-1}4^{2j}+	(\sum_{\substack{j=0}}^{n-1}4^{2j}+4^n\cdot\sum_{\substack{j=0}}^{n-1}4^{2j})\pmod{4^{2n}-1}\\
		&=\left\{ \begin{array}{ll}
			2\cdot 2f+2\cdot 2f+(4f+1)+(4f+1)\pmod{3}\\
			2\cdot 2f-2\cdot 2f\pmod{5}
		\end{array} \right.\\
		&=\left\{ \begin{array}{ll}
			2f+2\pmod{3}\\
			0\pmod{5}
		\end{array} \right..
	\end{align*}
Combining Lemmas \ref{lem5.1} and \ref{lem5.2} we prove the theorem.
		\end{proof}
	
		\begin{theorem}\label{th4.2}
		Let $t_i\ (1\leq i\leq 6)$ be six binary sequences of length $n=4f+1$ with support sets $D_0\cap D_1$, $D_0\cap D_2$, $D_0\cap D_3$, $D_1\cap D_2$, $D_1\cap D_3$, $D_2\cap D_3$, respectively, and $s=s(a,b)$ be the optimal quaternary sequence given by Lemma \ref{lem8} or Lemma \ref{lem9}, then the 4-adic complexity of the sequence $s$ satisfies
		\begin{align*}
			C_{4}(s)
			\geq\log_{4}(\frac{4^n+1}{5}).
		\end{align*}
	\end{theorem}
	\begin{proof}
		Taking $e=(0,0,0)$ and $(c^0,c^1,c^2,c^3)=(t_1,t_2,t_2,t_1)$ in Lemma \ref{lem8} as an example, from (\ref{ES5}) we know that
		
		\begin{align}\label{19}
			S(4)
		=&2\sum_{\substack{j=0\\t_1(j)=1}}^{n-1}4^{2j}+2\cdot 4^n\cdot\sum_{\substack{j=0\\t_2(j)=1}}^{n-1}4^{2j}+(1+4^n)\sum_{\substack{j=0\\t_1(j)+t_2(j)=1}}^{n-1}4^{2j}\pmod{4^{2n}-1}\notag\\
			=&2\sum_{\substack{j=0\\j\in D_0\cup D_1}}^{n-1}4^{2j}+2\cdot 4^n\cdot\sum_{\substack{j=0\\j\in D_0\cup D_2}}^{n-1}4^{2j}+(1+4^n)\sum_{\substack{j=0\\j\in D_1\cup D_2}}^{n-1}4^{2j}\pmod{4^{2n}-1}\\
			=&2(\eta_0+\eta_1)+2\cdot 4^n(\eta_0+\eta_2)+(1+4^n)(\eta_1+\eta_2)\pmod{4^{2n}-1}\notag\\
			=&\left\{ \begin{array}{ll}
				2(\eta_0+\eta_1+\eta_0+\eta_2+\eta_1+\eta_2)\pmod{4^{n}-1}\\
				2(\eta_0+\eta_1-\eta_0-\eta_2)\pmod{4^{n}+1}
			\end{array}\right.\notag\\
			=&\left\{ \begin{array}{ll}
				4(\sum\limits_{j=1}^{n-1}4^{2j}-\eta_3)\pmod{4^{n}-1}\\
				2(\eta_1-\eta_2)\pmod{4^{n}+1}
			\end{array}\right.\notag\\
			=&\left\{ \begin{array}{ll}
				-4(\eta_3+1)\pmod{\frac{4^{n}-1}{3}}\\
				2(\eta_1-\eta_2)\pmod{\frac{4^{n}+1}{5}}
			\end{array}\right.\notag
		\end{align}
	and by (\ref{19}) we have
	\begin{align*}
			S(4)=\left\{ \begin{array}{ll}
				0\pmod{3}\\
				0\pmod{5}
			\end{array} \right..
	\end{align*}
Combining with Lemma \ref{lem5.2}, we have $\gcd(S(4), 4^{2n}-1)\leq 5\cdot(4^{n}-1)$, the rest of the cases are similar and the same conclusion can be drawn.
\end{proof}
\begin{example}
	Let $n=13=12+1=4\cdot1^2+3^2,\ y=-1,$ then $\mathbb{F}_{13}^{*}=\langle2\rangle$. The cyclotomic classes of order 4 in  $\mathbb{F}_{13}^{*}$ are $D_0=\{1,\ 3,\ 9\},\ D_1=\{2,\ 5,\ 6\},\ D_2=\{4,\ 10,\ 12\},\ D_3=\{7,\ 8,\ 11\}$. It is easy to see that 
	\begin{align*}
		t_1=(0111011001000),\\
		t_2=(0101100001101),\\
		t_6=(0000100110111).
	\end{align*}
Let $e=(0,0,0)$, $(c^0,c^1,c^2,c^3)=(t_2,t_1,t_2,t_1)$,  then we have
\begin{align*}
	a=b=(00100110100000010111100111),\\
	s=(00200220200000020222200222).
\end{align*}
Since
$\gcd(S(4),\ 4^{26}-1)=\gcd(2^5\cdot3^2\cdot5\cdot17\cdot31\cdot67\cdot58184921,\ 3\cdot5\cdot53\cdot157\cdot1613\cdot2731\cdot8191)=15$, we have $C_s(4)=\log_4(\frac{4^{26}-1}{15}).$ From Theorem \ref{th4.1} we know that $d_1=\gcd(13^2+3\cdot13+4,\ \frac{4^{13}-1}{3})=1$, $3\mid f=12$, $C_s(4)=\log_4(\frac{4^{2n}-1}{15d_1})=\log_4(\frac{4^{26}-1}{15}),$ which is consistent with the direct calculation.
\end{example}

\section{Conclusion}\label{sec5}
Quaternary sequences with optimal autocorrelation have an important role in communication and cryptography systems. Su et al. \cite{S1} constructed several new familis of optimal autocorrelation quaternary sequences by using interleaved construction, sequences pairs and binary cyclotomic sequences of order four. In this paper, firstly, we determine the 4-adic complexity of quaternary sequences interleave by a pair of twin-prime sequences or GMW sequences with correlation fuction of the two pairs binary sequences. Secondly, we introduce the definition of the ``Gauss periods" of order four and ``quadratic Gauss sums" on finite field $\mathbb{F}_n$ and  valued in  $\mathbb{Z}^{*}_{4^{2n}-1}$, and calculate the 4-adic complexity of interleaved quaternary sequences constructed by two or three binary cyclotomic sequences of order four. Our results show that the 4-adic complexity of these sequences is larger than $\frac{2n-16}{6}$ and they are safe enough to resist the attack of the rational approximation algorithm.
	

\begin{thebibliography}{99}
		\bibitem{B. C. Barndt}
		Berndt B C, Evans R J, Williams K S. Gauss and Jacobi Sums. John Wiley and Sons INC, 1998
		\bibitem{EA.} Edemskiy V, Chen Z X. On the 4-adic complexity of the two-prime quaternary generator. J. Appl. Math . Comput. 2022, http://doi.org/10.1007/s12190-200-01740-z
		\bibitem{J}  Jang J-W, Kim Y-S, Kim S-H, et al. New quaternary sequences with ideal autocorrelation constructed from binary
		sequences with ideal autocorrelation. In: Proceedings of IEEE International Symposium on Information Theory, Seoul,
		2009. 278-281
		\bibitem{Jing} Jing X Y, Xu Z F, Yang M H, et al. On the p-Adic Complexity of the Ding-Helleseth-Martinsen binary sequecnes. Chin J Electron, 2021, 30: 64-71
		\bibitem{Kim2} Kim Y-S, Jang J-W, Kim S-H, et al. New construction of quaternary sequences with ideal autocorrelation from
		Legendre sequences. In: Proceedings of the IEEE international conference on Symposium on Information Theory,
		Seoul, 2009. 282-285
		\bibitem{Kim1} Kim Y-S, Jang J-W, Kim S-H, et al. New quaternary sequences with optimal autocorrelation. In: Proceedings of the IEEE International Conference on Symposium on Information Theory, Seoul, 2009. 286-289
		\bibitem{K1}
		Klapper A. A survey of feedback with carry shift registers. In: Proceedings of Sequences and Their Applications, Seoul, 2004. 56-71
		\bibitem{K2}
		Klapper A, Xu J Z. Register synthesis for algebraic feedback shift registers based
		on non-primes. Des Codes Cryptogr, 2004, 31: 227-250
		\bibitem{L} L\"{u}ke H D, Schotten H D,  Hadinejad-Mahram H. Generalised Sidelnikov sequences with optimal autocorrelation properties. Electron Lett, 2000, 36: 525-527
		\bibitem{Q1}Qiang S Y, Li Y, Yang M H, et al. The 4-Adic Complexity of A Class of Quaternary Cyclotomic Sequences with Period 2p. arXiv:2011.11875
		\bibitem{Q} Qiang S Y, Jing X Y, Yang M H,  et al. 4-Adic Complexity of Interleaved Quaternary Sequences. arXiv:2105.13826
		\bibitem{T. Storer}
		Storer T. Cyclotomy and Difference Sets. Markham, Chicago, 1967
		\bibitem{S1} Su W, Yang Y, Zhou Z C, et al. New quaternary sequences of even length with
		optimal auto-correlation. Sci China Inf Sci, 2018, 61: 1-13
		\bibitem{T1}Tang X H, Ding C S. New classes of balanced quaternary and almost balanced binary sequences with optimal autocorrelation value. IEEE Trans Inf Theory, 2010, 56: 6398-6405
		\bibitem{T2}Tang X H, Gong G. New constructions of binary sequences with optimal autocorrelation value/magnitude. IEEE Trans
		Inf Theory, 2010, 56: 1278-1286
			\bibitem{Y} Yang M H, Qiang S Y, Jing X Y, et al. On the 4-Adic Complexity of Quaternary Sequences with Ideal Autocorrelation. In: Proceedings of IEEE International Symposium on Information Theory, Espoo,
		2022. 528-531
		\bibitem{Zhang}
		Zhang L L, Zhang J, Yang M H, et al. On the 2-adic complexity of the Ding-Helleseth-Martinsen binary sequecnes. IEEE Trans
		Inf Theory, 2020, 66: 4613-4620
	\end{thebibliography}
\end{document}